\RequirePackage{amsmath}
\documentclass[runningheads]{llncs}
\usepackage{amssymb}
\usepackage{verbatim}
\usepackage{gastex}
\usepackage{algorithmicx,algorithm,algpseudocode}
\usepackage{pstricks,pstricks-add,pst-node}
\usepackage{geometry,array,graphicx,multicol}
\usepackage{xcolor}
\usepackage[lowtilde]{url}

\usepackage[utf8]{inputenc}
\usepackage[T1]{fontenc}
\DeclareSymbolFont{rsfscript}{OMS}{rsfs}{m}{n}
\DeclareSymbolFontAlphabet{\mathrsfs}{rsfscript}

\newcommand{\LineIf}[2]{\State \algorithmicif\ {#1}\ \algorithmicthen\ {#2} \algorithmicend\ \algorithmicif}
\newcommand{\Q}{\mathbb{Q}}

\begin{document}
\title{Generating Synchronizing Automata\\with Large Reset Lengths}
\author{Andrzej Kisielewicz\thanks{Supported in part by Polish MNiSZW grant IP 2012 052272.}
\and Marek Szyku{\l}a\thanks{Supported in part by Polish NCN grant 2013/09/N/ST6/01194.}}
\institute{Department of Mathematics and Computer Science, University of Wroc{\l}aw
\email{andrzej.kisielewicz@math.uni.wroc.pl,\ msz@ii.uni.wroc.pl}}
\maketitle

\begin{abstract}
We study synchronizing automata with the shortest reset words of relatively large length. First, we refine the Frankl-Pin result on the length of the shortest words of rank $m$, and the B\'eal, Berlinkov, Perrin, and Steinberg results on the length of the shortest reset words in one-cluster automata. The obtained results are useful in computation aimed in extending the class of small automata for which the \v{C}ern\'y conjecture is verified and discovering new automata with special properties regarding synchronization.
\end{abstract}


\section{Introduction}

We deal with deterministic finite automata $\mathcal{A} = \langle Q,\Sigma,\delta \rangle$, where $Q$ is the set of the states, $\Sigma$ is the input alphabet, and $\delta\colon \; Q \times \Sigma \to Q$ is the (complete) transition function. The cardinality $n=|Q|$ is the \emph{size} of $A$, and if $k=|\Sigma|$ then $\mathcal{A}$ is called $k$-\emph{ary}. The \emph{rank} of a word $w \in \Sigma^*$ is $|Qw|$, and the \emph{rank} of $\mathcal{A}$ is the minimal rank of a word over $\mathcal{A}$. For a nonempty subset $\Sigma'\subseteq \Sigma$, we may define the automaton $\mathcal{A}' = \langle Q,\Sigma',\delta' \rangle$, where $\delta'$ is the natural restriction of $\delta$ to $\Sigma'$. In such a case $\mathcal{A}$ is called an \emph{extension} of $\mathcal{A}'$. 
The automata of rank 1 are called \emph{synchronizing}, and each word $w$ with $|Qw|=1$ is called a \emph{synchronizing} (or \emph{reset}) word for $\mathcal{A}$.

We are interested in the length of a shortest reset word for $\mathcal{A}$ (there may be more than one word of the same shortest length). We call it the \emph{reset length} of $\mathcal{A}$. 
The famous \v{C}ern\'y conjecture states that every synchronizing automaton $\mathcal{A}$ with $n$ states has a reset word of length $\leq (n-1)^2$. This conjecture was formulated by \v{C}ern\'y in 1964 \cite{Cerny1964}, and is considered the longest-standing open problem in combinatorial theory of finite automata. So far, the conjecture has been proved only for a few special classes of automata and a cubic upper bound has been established (see Volkov \cite{Volkov2008Survey} for an excellent survey). In~\cite{KS2013GeneratingSmallAutomata} we have verified the conjecture for all binary automata with $n<12$ states.

In this paper we prove some new results improving known bounds and extending the class of automata for which the \v{C}ern\'y conjecture is verified. In particular, we strengthen the Frankl-Pin result on the length of the shortest words of rank $m$, and the B\'eal, Berlinkov, Perrin~\cite{BP2009QuadraticUpperBoundInOneCluster,BBP2011QuadraticUpperBoundInOneCluster}, and Steinberg~\cite{Steinberg2011OneClusterPrime} results on one-cluster automata. These are refinements of a rather technical nature. The motivation for these refinements is to make computations in this area more effective.

This allows to extend the studies reported in~\cite{Tr2006Trends,AGV2010,AGV2013,KS2013GeneratingSmallAutomata}.
In particular, we search for synchronizing automata with relatively large reset length. We improve the algorithm from \cite{KS2013GeneratingSmallAutomata} which takes a set of $(k-1)$-ary automata with $n$ states and generates all their nonisomorphic one-letter extensions. To perform an exhaustive search over the $k$-ary automata with $n$ states with some property, we need to progressively run the algorithm $k-1$ times starting from the complete set of non-isomorphic unary automata. However, in each run, if we know that any extension of an automaton $\mathcal{A}$ cannot have the desired property, we can safely drop $\mathcal{A}$ from further computations. Since the number of generated automata grows rapidly, suitable knowledge saves a lot of computational time and extends the class of the automata investigated. In this study, we concentrate on automata of arity $k>2$.


\section{Theoretical Base for Computation}\label{sec:bounds}

Through the paper, if not stated otherwise, $\mathcal{A}$ denotes a (deterministic, finite) automaton, $\Sigma$ its alphabet, $Q$ its set of the states, and $n=|Q|$ its size. Words are always the words over $\Sigma$ (that is, elements of $\Sigma^*$). A word $w$ is said to \emph{compresses} a set $M \subseteq Q$ if $|Mw| < |M|$. In such a case $M$ itself is called \emph{compressible}.

\subsection{Frankl-Pin sequences}

Suppose that $M \subseteq Q$ and $u=a_1 \dots a_\ell$ is a shortest word compressing $M$, that is, such that $|Mu| < |M|$. Let $M_i = M a_1 \dots a_{i-1}$ for $1 \le i \le \ell+1$. Then, there are $x_\ell,y_\ell \in M_\ell$ such that $x_\ell a_\ell = y_\ell a_\ell$. For $\ell > i \ge 1$, we define $x_i, y_i\in M_i$  by $x_ia_i = x_{i+1}$ and $y_ia_i = y_{i+1}$. This defines $x_i$ and $y_i$ uniquely, since otherwise $u$ would not be a shortest word compressing $M$. By the same reason the subsets $M_i$ are of the same cardinality for $i \le \ell$, and together with the pairs $R_i = \{x_i,y_i\}$ they satisfy the following conditions:

\begin{enumerate}
\item $R_i \subseteq M_i$ for $1 \le i \le \ell$;
\item $R_i \not\subseteq M_j$ for $1 \le j < i \le \ell$.
\end{enumerate}

In~\cite{Pin1983OnTwoCombinatorialProblems}, J.-E. Pin used this observation to bound the length of a word $u$ compressing $M$. He suggested a certain combinatorial estimation that was proved subsequently by Frankl. (We quote the result in a restricted form sufficient for our aims).

\begin{theorem}\label{thm:frankl}\textnormal{(P.~Frankl \cite{Fr1982})}
Let $Q$ be an $n$-element subset, $M_1,\ldots,M_\ell$ be a sequence of its $m$-subsets (for some $1 < m \le n$), and $R_1,\ldots,R_\ell$ be a sequence of pairs contained in $Q$. If the conditions~1 and 2 above are satisfied, then $$\ell \le \binom{n-m+2}{2}.$$
\end{theorem}

We say that a sequence $(M_i,R_i)$, $(1 \le i \le \ell)$ of $m$-subsets $M_i$ and pairs $R_i$ satisfying conditions~1 and~2
is an \emph{$m$-subset Frankl-Pin sequence}. If all the pairs $R_i$ belong to a set $P$ of pairs, we will say that this sequence is \emph{over} $P$. From what we said it follows that a shortest word compressing $M$ cannot be longer than the length of the Frankl-Pin sequence starting from $M$. 
Hence summing up the binomial coefficient in Theorem~\ref{thm:frankl} we obtain, in particular, the bound $(n^3-n)/6$ for the length of a shortest reset word. In spite of many efforts to improve it, it is still the best bound known in the literature.  

In order to make a slight technical improvement, we introduce the following notions. Let $P$ be an arbitrary set of compressible pairs in $\mathcal{A}$. By a \emph{synchronizing height} $h(P)$ of $P$ we mean the minimal $h$ such that for each pair $\{x,y\}\in P$ there exists a word $w$ of length $h$ such that $xw=yw$. We make use of the observation that if the synchronizing height is smaller than the maximal length of a Frankl-Pin sequence over $P$, then we can improve the Pin's estimation from \cite{Pin1983OnTwoCombinatorialProblems}.

\begin{theorem}\label{thm:pairs_bound}
Let $P$ be a set of compressible pairs in $\mathcal{A}$, $h(P)$ the synchronizing height of $P$, and $p(P,m)$ the maximal length of an $m$-subset Frankl-Pin sequence over $P$. 
Then, for every compressible $m$-subset $M$ of $Q$ ($2 \le m \le n$), there is a word compressing $M$ whose length does not exceed 
$$\binom{n-m+2}{2} - p(P,m) + h(P).$$
\end{theorem}
\begin{proof}
(In the first part of the proof we modify the Pin's argument mentioned above; see \cite[Proposition~3.1]{Pin1983OnTwoCombinatorialProblems}).
Let $u=a_1 \dots a_\ell$ be a shortest word such that either $|Mu| < |M|$ or $\{x_k,y_k\} \subseteq Mu$ for some $\{x_k,y_k\} \in P$. First observe, that if $|u|=0$, then it means that $M$ contains a pair from $P$, and consequently, there is a word $w$ compressing $M$ of length $|w| \le h(P)$. Since, by Theorem~\ref{thm:frankl}, $\binom{n-m+2}{2} \ge p(P,m)$, $w$ has the required length. 

Thus, we may assume that $|u| \ge 1$. 
Let $M_i = M a_1 \dots a_{i-1}$ for $1 \le i \le \ell+1$. 
Since $|M_{\ell+1}| < |M_\ell|$ or there are $\{x,y\} \in M_{\ell+1}$ with $\{x,y\} \in P$, we have that $M_\ell$ contains two distinct states $x_\ell,y_\ell$ such that either $x_\ell a_\ell = y_\ell a_\ell$, or $x_\ell a_\ell = x$ and $y_\ell a_\ell = y$. For $\ell > i \ge 1$
we define $R_i = \{x_i,y_i\} \subseteq M_i$ by $x_i a_i = x_{i+1}$ and $y_i a_i = y_{i+1}$.  
The sequence $(M_i,R_i)$ is a Frankl-Pin sequence. Indeed, condition~1 holds by definition. For condition~2, assume that $R_i \in M_j$ for some $1 \le j < i \le \ell$. Then, for word $u'=a_1 \dots a_{j-1} a_i \dots a_\ell$ we have a pair of distinct states $x',y' \in M$ such that either $x'u'=y'u'$ or $\{x'u',y'u'\} \in P$, and $u'$ is shorter than $u$, which is a contradiction.

We extend the sequence $(M_i,R_i)$ as follows. Let $(T_i,P_i)$ be an $m$-subset Frankl-Pin sequence over $P$ with $1\le i \le p(P,m)$. Define
\begin{itemize}
\item $M'_i = M_i$ and $R'_i = R_i$ for $1 \le i \le \ell$;
\item $M'_i = T_{i-\ell}$ and $R'_i = P_{i-\ell}$ for $\ell+1 \le i \le \ell+p(P,m)$.
\end{itemize}
Then $(M'_i,R'_i)$ is a Frankl-Pin sequence of length $\ell + p(P,m)$. Indeed, 
condition~1 trivially holds. For condition~2 it is enough to observe that, by the definitions of $M_i$ and $(T_i,P_i)$, for $i > \ell$, $R'_i = P_{i-\ell}$ is not in any $M'_j$ with $j < i$.

Now, by Theorem~\ref{thm:frankl}, $\ell+p(P,m) \le \binom{n-m+2}{2}$. Thus $|u| \le \binom{n-m+2}{2} - p(P,m)$. If $|Mu| < |M|$ we are done. Otherwise, $\{x',y'\} \in Mu$ for some $\{x',y'\} \in P$ and we must append to $u$ a word compressing $\{x',y'\}$, which has length at most $h(P)$. As a result we obtain a word $w$ of length at most $\binom{n-m+2}{2} - p(P,m) + h(P)$, as required.
\qed
\end{proof}

Our result is to be applied in concrete situations, when we can find a relatively large set of compressible pairs with small synchronizing height. In order to estimate the minimal length of Frankl-Pin sequence we have the following auxiliary result.  
Given words $w_1, \ldots, w_k$ we define $P(w_1, \ldots, w_k)$ to be the set $P$ of pairs $(x,y)$ such that $xw_i = yw_i$ for some $1 \le i \le k$. Given $k$, choose words $w_1,\ldots, w_k$ so that $P(w_1, \ldots, w_k)$ is of maximal cardinality. Denote this cardinality by $p(k)$.

\begin{proposition}\label{bounds-pro:len_sequence}
If $\mathcal{A} = \langle Q,\Sigma,\delta \rangle$ is an $n$-state automaton of rank $r$, 
then  
for each $2 \le m \le r$ there exists a Frankl-Pin sequence of $m$-subsets of length $p = p(\lfloor (r-m+3)/2 \rfloor)$.
\end{proposition}
\begin{proof}
Fix $m \le r$, and let $c=\lfloor (r-m+3)/2 \rfloor$. Choose $w_1,\ldots,w_c$ so that
$P=P(w_1, \ldots, w_c)$ has cardinality $p=p(c)$. Fix some ordering of pairs in $P$ denoting
$P=(P_i)$, ${1 \le i \le p}$, $P_i = \{x_i,y_i\}$. We proceed to define corresponding $(M_i)$. To this end, let $w$ be a word of rank $r$, and let $S = Qw$. Then $S$ is not compressible, and in particular,  no pair $(x,y)\in P$ is contained in $S$.

Given a pair $P_i = \{x_i,y_i\}$, let $T_i$ consists of all elements $x\in S$ such that $xw_k = x_iw_k$ or $xw_k = y_iw_k$ for some $1\le k\le c$. Note that, since $S$ is not compressible,  for every $k$, there is at most one $x$ such that $xw_k = x_iw_k$. Similarly, there is at most one $x$ such that $xw_k = y_iw_k$. Moreover, since $x_iw_k = y_iw_k$ for some $k$, the cardinality $|T_i| \le 2c-1$, and consequently the $|S\setminus T_i| \ge r-2c+1 \ge m-2$.
 
We define $M_i = S'_i \cup  \{x_i,y_i\}$, where $S'_i $ is an arbitrary $(m-2)$-subset of $S\setminus T_i$. 
Now, consider a pair $P_j =\{x_j,y_j\}$ with $j \neq i$. 
Since $S'_i$ is not compressible, $P_j \not\subseteq S'_i$. Since $P_j\neq P_i$, the remaining case for $P_j \subseteq M_i$ is when  $|P_j \cap S'_i| = 1$ and one of $x_i$ or $y_i$ belongs to $P_j$. Then, take $k$ with $x_jw_k = y_jw_k$. It follows that there is $x\in S'_i$ such that either $xw_k = x_iw_k$ or $xw_k = y_iw_k$, which contradicts the fact that  $S'_i \subseteq S\setminus T_i$. 
Consequently, $P_j \not\subseteq M_i$. This proves that $(P_i,M_i)$ is a Frankl-Pin sequence over $P$ of required length.
\qed
\end{proof}

In the case of $|\Sigma|=1$ we have the following more specific result, which may be used to estimate the reset length, when one letter of the automaton is known.

\begin{corollary}\label{bounds-cor:1_letter_sequence}
If $\mathcal{A}$ is a unary automaton of rank $r$, and $P$ is the set of all compressible pairs in $\mathcal{A}$, then for each $2 \le m \le r$ there exists an $m$-subsets Frankl-Pin sequence over $P$ of length $p(P,m) = |P|$. Moreover, $|P| \ge \frac{1}{2} n(\frac{n}{r} - 1)$, and $h(P) = n-r$.
\end{corollary}
\begin{proof}
We have $\Sigma =\{a\}$ in this case, and each word is of the form $w=a^h$. It follows that for each set $w_1,\ldots,w_k$, there is $h \le n-r$ such that  $P(w_1,\ldots,w_k) = P(a^h)$. Note that $a^{n-r}$ compresses each compressible pair. 
Consequently, taking $w_1=a^{n-r}$, we see that for every $k$, $p(k)$ is the cardinality of the set $P$ of all compressible pairs, and the first part of the result follows from Proposition~\ref{bounds-pro:len_sequence}. The set $P$, in this case, is the union of equivalence classes of the relation determined by the letter $a$ with the condition: $q, p \in Q$ are in the relation if and only if $qa^{n-r} = pa^{n-r}$. If $c_1,\ldots,c_r$ are the cardinalities of the equivalence classes, then $|P| \ge \sum_{i=1}^r \frac{c_i(c_i-1)}{2}$. The expression achieves its minimum when all $c_i$ are equal. Hence, $|P| \ge r\frac{n}{2r}(\frac{n}{r} - 1)$, as claimed. Obviously, $h(P)$ is given by the fact that $a^{n-r}$ compresses each compressible pair. 
\qed
\end{proof}

Of course, the bound above is the worst case. Having a concrete letter $a$ one may compute the exact value of $|P|$, which in particular cases may be as large as $\binom{n-r+1}{2}$.

Theorem~\ref{thm:pairs_bound} should be combined and compared with the following result proved by J.-E. Pin:

\begin{theorem}\label{thm:pin_rank_bound}\textnormal{(J.-E. Pin \cite{Pin1972Utilisation})}
Let $\mathcal{A} = \langle Q,\Sigma,\delta \rangle$ be an automaton with $n=|Q|$ states, and $u \in \Sigma^*$ be a word of rank $m$. If there is a word of rank $\le m-1$, then there is such a word of length at most $2|u| + n - m + 1$.
\end{theorem}

In our computation, given an automaton $\mathcal{A}$ of rank $m$, we can bound the reset length of a synchronizing extension by applying successively $m-1$ times either Theorem~\ref{thm:pairs_bound} or Theorem~\ref{thm:pin_rank_bound}, depending on which bound is smaller for the given rank. It can be demonstrated that the best results are achieved when each time we take a word of the minimal rank $m$, and check which of the propositions gives a better bound. It turned out that first, for larger ranks, Theorem~\ref{thm:pin_rank_bound} gives a better bound, and then Theorem~\ref{thm:pairs_bound} becomes more effective. Also note that, when applied to get a bound on the length of a shortest reset word, the Pin's result gives an exponential estimate, while our result gives a polynomial bound.


\subsection{One-cluster automata}

A very useful result from the computational point of view is the result of Steinberg~\cite{Steinberg2011OneClusterPrime} on one-cluster automata. 
Recall that an automaton $\mathcal{A} = \langle Q,\Sigma,\delta \rangle$ is \emph{one-cluster}, if it has a letter $a\in\Sigma$ such that for every pair $q,s \in Q$ there are $i,j \ge 1$ such that $qa^i = sa^j$. This means that the graph of the transformation induced by $a$ is connected. In particular, it has a unique cycle $C \subseteq Q$ with the property $Ca^i = C$ for every $i\ge 0$, and there is $\ell\ge 0$ such that $Qa^\ell = C$. The least such $\ell$ is called the \emph{level}\index{level} of $\mathcal{A}$. 
Steinberg \cite{Steinberg2011OneClusterPrime} proved that if the length $m$ of the cycle is prime, then the one-cluster automaton $\mathcal{A}$ has a reset word of length at most
\begin{equation}\label{bounds-eq:prime}
n - m + 1 + 2\ell + (m - 2)(n + \ell).
\end{equation}
We generalize this result to arbitrary lengths.

There is a series of results establishing a general quadratic upper bound for the reset length of one-cluster automata 
\cite{BP2009QuadraticUpperBoundInOneCluster,BBP2011QuadraticUpperBoundInOneCluster,Steinberg2011OneClusterPrime,Steinberg2011AveragingTrick,CarpiDAlessandro2013IndependendSetsOfWords}. For small lengths $m$ of the cycle, the best bound
\begin{equation}\label{bounds-eq:Stainberg}
2nm-3n-4m+2\ell+8
\end{equation}
was announced in Steinberg~\cite{Steinberg2011OneClusterPrime} (with a sketch of the proof). For larger $m$, the best bound 
\begin{equation}\label{bounds-eq:CarpiD'Alessandro}
2nm-2m\ln\frac{m+1}{2} - n - m
\end{equation}
has been obtained recently by Carpi and D'Alessandro~\cite{CarpiDAlessandro2013IndependendSetsOfWords}.
In the proof of (\ref{bounds-eq:CarpiD'Alessandro}) there is a more complicated formula that improves (\ref{bounds-eq:Stainberg}) for all $m$. Our generalization improves all these bounds. Note also that, in contrast with our result, the mentioned bounds for $m$ prime are weaker than the Steinberg's bound (\ref{bounds-eq:prime}).

From the computational point of view it is good to have a bound involving all possible parameters $n,m,\ell$. Our idea is to refine the proof of Steinberg \cite{Steinberg2011OneClusterPrime} by making two essential modifications. First, we formulate and prove a different crucial lemma that does not employ the fact that the length of the cycle is prime. At second, we estimate more precisely the dimension of the vector spaces involved, which gives a better estimation of the length of the resulted reset word.

First we need to recall basic notations from~\cite{Steinberg2011OneClusterPrime}. We consider the matrix representation
$\pi\colon \Sigma^* \to M_n(\Q)$ defined by $\pi(w)_{q,r} = 1$ if $qw=r$, and $0$, otherwise. Given $S\subseteq Q$ we define $[S]$ to be the characteristic row vector\index{characteristic vector} of $S$ in $\Q^n$, $[S]^T$ its transpose, and 
$$\gamma_S = [S]^T - (|S|/|C|)[Q]^T.$$
By $[C]w\gamma_S$ we denote the product of corresponding matrices; in particular, 
the word $w$ represents in this notation the matrix $\pi(w)$, and the whole product is an element of $\Q$. 
Further, for any word $w\in \Sigma^*$, $w\gamma_S$ is the vector obtained as the product of the transformation matrix corresponding to $w$ by the vector $\gamma_S$.

In~\cite{Steinberg2011OneClusterPrime} (and earlier papers), the following fact is used
\begin{equation}\label{bounds-eq:main}
[C]w\gamma_S = |C\cap Sw^{-1}| - |S|.
\end{equation}
If this difference is larger than zero then the preimage of $S$ by the word $w$ has more elements in the cycle $C$ than $S$ itself, and in consequence $w$ compresses $C$. So, in general, the approach is in looking for short words $w$ for which $[C]w\gamma_S > 0$.
We will be interested in the subspace $W_S = \mbox{\rm Span}\{a^{\ell+j}\gamma_S \in \Q^n \;|\; 0\le j \le m-1\}$ (cf.~\cite{Steinberg2011OneClusterPrime}).

In addition, we introduce the number $D(m,k)$ as follows. Let $[S] = (a_1,\ldots,a_m)$, $a_i\in \{0,1\}$. By the \emph{cyclic transforms}\index{cyclic transform} of $[S]$ we mean the following vectors:
$(c_1,\ldots,c_m)$, $(c_2,\ldots,c_m,c_1)$, $\ldots$, $(c_m,c_1,\ldots,c_{m-1})$. Given $1\le k \le m$, by $D(m,k)$ we denote the minimal dimension of the subspace $V_S$ of $\Q^m$ generated by the cyclic transforms of a vector $[S]$ with $|S|=k$ (that is, $[S]$ runs here over all vectors with exactly $k$ ones and $m-k$ zeros). Obviously, $D(m,1)=D(m,m-1)=m$. Yet, for example, $D(2k,k)=2$. More information about $D(m,k)$ can be inferred from~\cite{Ingleton1956RankOfCirculantMatrices}, where in particular the rank of the matrix generated by the cyclic transforms of a vector is considered. 

We can prove the following

\begin{lemma}\label{bounds-lem:main_dim}
Let $\mathcal{A}$ be a synchronizing one-cluster automaton with level $\ell \geq 0$ and cycle $C$ of length $m>1$.
Let $S$ be a subset of $C$ of cardinality $|S|=k>0$, and $W_S = \mbox{\rm Span}\{a^{\ell+j}\gamma_S \in \Q^n \;|\; 0\le j \le m-1\}$. Then $\dim\;W_S \ge D(m,k)-1$ and the sum of the generators
\begin{equation}\label{bounds-eq:zero}
\sum_{0\le j\le m-1} a^{\ell+j} \gamma_S = 0.
\end{equation}
\end{lemma}
\begin{proof}
First note, that by definition, 
$$a^{\ell+j}\gamma_S = a^{\ell+j}[S]^T - a^{\ell+j}\frac{k}{m}[Q]^T.$$ 
The summands are (as it is easy to check; see \cite{Steinberg2011OneClusterPrime}) the characteristic vectors of preimages 
$$[S(a^{\ell+j})^{-1}]^T - \frac{k}{m}[Q(a^{\ell+j})^{-1}]^T = [S(a^{\ell+j})^{-1}]^T - \frac{k}{m}[Q]^T.$$ 

Hence, 
\begin{equation}\label{bounds-eq:zero_sum}
\sum_{0\le j\le m-1} a^{\ell+j} \gamma_S = \left(\sum_{0\le j\le m-1}[S(a^{\ell+j})^{-1}]^T\right) - {k}[Q]^T.
\end{equation}
In order to compute the sum on the right hand side, we observe that, for each $q\in C$,
$$\sum_{0 \le j \le m-1} [q(a^{\ell+j})^{-1}]^T = [Q]^T .$$
This is so, because for every $q\in Q$ and $s\in C$, there is ${0\le j\le m-1}$ such that $qa^{\ell+j} = s$, and for all ${0\le i, j\le m-1}$, and $q\in Q$, if $qa^{\ell+j}=qa^{\ell+i}$ then $i=j$.
It follows that 
$$\sum_{0\le j\le m-1}[S(a^{\ell+j})^{-1}]^T = \sum_{0\le j\le m-1}\sum_{q\in S}[q(a^{\ell+j})^{-1}]^T = \sum_{q\in S}[Q]^T = k[Q]^T.$$
Combining this with (\ref{bounds-eq:zero_sum}) yields (\ref{bounds-eq:zero}).

It remains to estimate the dimension of $W_S$. Let us denote $w_j = [S(a^{\ell+j})^{-1}]^T$, and for $c\in\Q$, $\hat c =  c[Q]^T$. Then,  $a^{\ell+j}\gamma_S = v_j - \hat c$, for $c = k/m$. 
We consider the restriction $V \subset \Q^m$ of $W_S$ to the coordinates corresponding to $[C]$, which formally is the image in the orthogonal projection $\phi$ of $W_S$ on the orthogonal complement of $[Q\setminus C]$. Then, of course, $\mbox{\rm dim}\;W_S \ge \mbox{\rm dim}\;V$, and it is enough to estimate $\mbox{\rm dim}\;V$ from below. 

Let $v_i = \phi(w_i) \in V$ be the image of $w_i \in W_S$ (i.e. restriction of $w_i$ to $m$ coordinates corresponding to $C$). Then $v_0$ is the characteristic vector of $S$ in $C$, and $v_0, \ldots, v_{m-1}$ are simply the cyclic transforms of $v$ in $\Q^m$. Consequently, for $U=\mbox{\rm Span} \{v_0, \ldots, v_{m-1} \}$,   $\mbox{\rm dim}\;U \geq D(m,k)$. Moreover, we have
$V = \mbox{\rm Span} \{v_0 - \bar d, \ldots, v_{m-1} - \bar d\}$, where $\bar d \in \Q^m$ denotes  $\bar d =  d[C]^T$ with $d=k/m$. Since $\sum_{0\le j \le m-1} v_j = \bar k$,
$$U = \mbox{\rm Span} \{v_0, \ldots, v_{m-1}, \bar d \} = \mbox{\rm Span} \{v_0 - \bar d, \ldots, v_{m-1} - \bar d, \; \bar d \} = \mbox{\rm Span} \{V, \bar d\}.
$$

Now, since the sum of the coordinates in each $v_i - \bar d$ is equal to $$k(1-k/m) - (m-k)k/m = 0,$$ it follows that $\bar d \notin V$.  Consequently
$$\mbox{\rm dim}\;V \ge  \mbox{\rm dim}\;U-1 \ge D(m,k)-1,$$
as required.
\qed
\end{proof}

Using this lemma in place of Lemma~4, Proposition~5, and Lemma~6 from \cite{Steinberg2011OneClusterPrime}, one can obtain a generalization of the Steinberg's bound (\ref{bounds-eq:prime}) with no assumption on the length of $C$. We can still generalize this result as follows. We observe that for $S\subseteq C$, 
$$W_S = \mbox{\rm Span}\{a^{\ell+j}\gamma_S \in \Q^n \;|\; 0\le j \le m-1\} = \mbox{\rm Span}\{a^{\ell+j}\gamma_S \in \Q^n \;|\; 0\le j \le q-1\},$$
where $q=q_S$ is the \emph{cyclic period} of $S$, understood as the least number $q$ such that $Sa^q=S$. In case when $S$ is not periodic on $C$, $q=m$, and nothing changes. But if $q<m$, then $q\leq m/2$, and $m$ in~(\ref{bounds-eq:zero}) may be replaced by $q$, giving better estimations of the reset length in the proof.\footnote{We are very grateful to Mikhail Berlinkov, who brought our attention to this fact.} Let us define ${D}^*(m,k)$ to be the minimal value of $m-q_S+\dim W_S$ taken over all vectors $S$ with $|S|=k$. Then we have the following:

\begin{theorem}\label{thm:one-cluster_bound}
Let $\mathcal{A} = \langle Q, \Sigma, \delta \rangle$ be a synchronizing automaton with $n$ states, such that there exists a word $w$ of length $s$ inducing a one-cluster transformation with level $\ell$ and cycle $C$ of length $m>1$.
Then $\mathcal{A}$ has a reset word of length at most
$$s(\ell +m-2)(m-1) + (n+1)(m-1)+s\ell - \sum_{k=1}^{m-1} {D}^*(m,k).$$
\end{theorem}
\begin{proof}
First we prove the result for $s=1$ and $D(m,k)$ in place of $D^*(m,k)$. We modify suitably the argument used in the proof in~\cite{Steinberg2011OneClusterPrime}. 

Let $S$ be a proper subset of $Q$ with $0< k=|S| < m$. We wish to show first that there exists a short word $w\in\Sigma^*$ with $[C]wa^{\ell+j}\gamma_S \neq 0$ for some $0\le j \le m-1$. If this holds for the empty word, we are done. Otherwise, $[C]a^{\ell+j}\gamma_S = 0$ for all $0\le j \le m-1$. This means $a^{\ell+j}\gamma_S \in [C]^\perp$ (the orthogonal complement of $[C]$). Since $\mathcal{A}$ is synchronizing, there exists a word $u$ resetting to a state in $C\cap S(a^{\ell})^{-1}$. Then, by~(\ref{bounds-eq:main}),
\begin{equation}\label{bounds-eq:last}
[C]ua^{\ell}\gamma_S = |C\cap S(ua^{\ell})^{-1}| - |S| = |C\cap Q| - |S| \neq 0.
\end{equation}
To find $u$ short enough with this property, let $W_S = \mbox{\rm Span}\{a^{\ell+j}\gamma_S \in \Q^n \;|\; 0\le j \le m-1\}$. We have $W_S \subseteq [C]^\perp$, yet by~(\ref{bounds-eq:last}), $uW_S \not\subseteq [C]^\perp$. 
By the standard ascending chain condition (see~\cite[Lemma~2]{Steinberg2011OneClusterPrime}; also cf.~\cite{Pin1972Utilisation,Dubuc1998,Kari2003Eulerian}), we infer that there exists a word $w$ satisfying $[C]wa^{\ell+j}\gamma_S \neq 0$ for some $j$, whose length $|w| \le \mbox{\rm dim}\;{C}^\perp - \mbox{\rm dim}\;W_S +1.$
By~(\ref{bounds-lem:main_dim}), we get $|w| \le n - D(m,k) +1.$

Moreover, by the same lemma, 
\begin{equation}\label{bounds-eq:sum}
\sum_{0\le j\le m-1}\!\![C]wa^{\ell+j}\gamma_S = [C]w\!\!\!\!\sum_{0\le j\le m-1}\!\!a^{\ell+j}\gamma_S = 0.
\end{equation}
Since $[C]wa^{\ell+j}\gamma_S \neq 0$ for some $j$, there must be $j$ such that $[C]wa^{\ell+j}\gamma_S > 0$. This means, by (\ref{bounds-eq:main}), that
$|C\cap S(wa^{\ell+j})^{-1}| - |S|  > 0,$
for some $0\le j\le m-1$ and 
$ |w| \le n - D(m,k) +1.$
Since $j\le m-1$, it means that there exists a word $u_k = w a^{\ell+j}$ of length at most $n+\ell+m-D(m,k)$ such that $|C\cap S(u_k^{-1})| > |S|$.

Thus, as in~\cite{Steinberg2011OneClusterPrime}, we may find a sequence of words $u_1,u_2,\ldots,u_{m-1}$ such that
starting from an arbitrary one-element set $S_1=\{q\}$, we have $|S_{k}| > |S_{k-1}|$ for $S_k =C\cap Su_k^{-1}$, $1\le k \le m-1$, and the length of $|u_k| \le n+\ell+m-D(m,k)$. In particular, the word $u=u_{m-1}\ldots u_2u_1$ has the property $Cu = \{q\}$, and since  $Qa^\ell = C$, the word $a^\ell u$ is synchronizing. Since the suffix of this word is $u_1 = wa^{\ell+j}$, we may replace $u_1$ by a shorter word $u'_1 = wa^\ell$, and  $v= a^\ell u_{m-1}\ldots u_2 u'_1$ is also synchronizing. For the length (using the inequalities obtained in the preceding paragraph) we have 
$$v \le \ell + (m-2)(n+\ell+m)  - \left(\sum_{2 \le k \le m-1} D(m,k)\right) + n-D(m,1) +1 +\ell.$$

Now we take into account the cyclic period $q=q_S$ of $S$. We observe that
$$W_S = \mbox{\rm Span}\{a^{\ell+j}\gamma_S \in \Q^n \;|\; 0\le j \le m-1\} = \mbox{\rm Span}\{a^{\ell+j}\gamma_S \in \Q^n \;|\; 0\le j \le q-1\}.$$
Moreover, the formula (\ref{bounds-eq:sum}) holds with $m$ replaced by $q$. Consequently, the same argument shows that there exists a word $u_k = w a^{\ell+j}$ of length at most $n+\ell+q-\mbox{\rm dim}\;W_S$ such that $|C\cap S(u_k^{-1})| > |S|$. Since $n+\ell+q-\mbox{\rm dim}\;W_S = n+\ell+m - (m-q+ \mbox{\rm dim}\;W_S)$, the length $|u_k| \leq n+\ell+m - D^*(m,k)$, and we get the result with $D(m,k)$ replaced by $D^*(m,k)$. 

Finally, to complete the proof let us assume that $\mathcal{A}$ has a one-cluster transformation corresponding to a word $v_a$. Let $\mathcal A'$ be the automaton obtained from $\mathcal{A}$ by adding to its alphabet an additional letter $a$ acting exactly as $v_a$. Applying the proof above to $\mathcal{A}'$, taking into account the places where the length $|v_a| = s$ counts, we obtain the required result.
\qed
\end{proof}

Let us see that it generalizes the mentioned Steinberg's result indeed.
Ingleton~\cite{Ingleton1956RankOfCirculantMatrices} showed that the dimension of the vector space generated by the transforms of a vector $[S] = (c_1,\ldots,c_m)$ is exactly $m-d$, where $d$ is the degree of the polynomial $g(x) = \gcd(c_1+c_2x+\ldots+c_mx^{m-1},\; x^m-1)$ in $\Q[x]$. In this case, when $m$ is prime, $g(x) = 1$, and $D^*(m,k)=m$ for all $k$. Substituting this in~Theorem~\ref{thm:one-cluster_bound}, for $s=1$, we obtain exactly the result (\ref{bounds-eq:prime}).

Our result is better then the bound (\ref{bounds-eq:CarpiD'Alessandro}) obtained by Carpi and D'Alessandro. The summands $D^*(m,k)$ in our formula, involving the greatest common divisors of polynomials, are not easy to estimate. But even a rough estimation $D^*(m,k) \geq \lceil m/k \rceil$ (as in~\cite[Lemma~4]{CarpiDAlessandro2013IndependendSetsOfWords}) yields the upper bound 
$
2nm-2m\ln\frac{m+1}{2} -2m - n + 1,
$
which is better by the summand $m$. A more sophisticated estimation yields the following 
\begin{corollary}\label{cor:one-cluster_estimation}
A synchronizing one-cluster automaton $\mathcal{A}$ with $n$ states and the cycle of length $m$ has a reset word of length at most
\begin{equation}
2nm-4m\ln\frac{m+3}{2}+2m-n+1
\end{equation}
\end{corollary}
\begin{proof}
First we show that $D^*(m,k) \ge 2m/(k+1)$. Let $S$ be a non-periodic subset of $C$ with $k$ states. Consider a linearly independent subset of the cyclic transforms of $[S]$ that generate $V_S$. Since $S$ is non-periodic, for each $i,j$ with $1 \le i \neq j \le m$ there is a cyclic transform of $[S]$ which has distinct values at $i$-th and $j$-th positions. So there are at most $\dim V_S$ positions at which exactly one vector contains 1. At the remaining $m-\dim V_S$ positions at least two vectors contain 1. Thus there are at least $\dim V_S+2(m-\dim V_S)$ ones in all the vectors in the base, so $k \dim V_S \ge \dim V_S+2(m-\dim V_S)$ and $\dim V_S \ge 2m/(k+1)$.

Now, following the estimation given in~\cite[Proposition~5]{CarpiDAlessandro2013IndependendSetsOfWords}, and using our bound $D^*(m,k) \ge 2m/(k+1)$, for odd $m$ we have
$$\sum_{k=1}^{m-1} D^*(m,k) \geq 2\sum_{k=1}^{(m-1)/2} D^*(m,k) \geq 4m\sum_{k=2}^{(m+1)/2} 1/k \geq 4m\ln\frac{m+3}{2} - 4m,$$
and for even $m$
\begin{eqnarray*}
\sum_{k=1}^{m-1} D^*(m,k) & \geq & D^*(m,m/2) + 2\sum_{k=1}^{(m-2)/2} D^*(m,k) \\
& \geq & \frac{2m}{m/2+1} + 4m\sum_{k=2}^{m/2} 1/k \\
& \geq & 4m/(m+2) -4m + 4m\ln\frac{m+2}{2} \\
& \geq & 4m\ln\frac{m+3}{2} - 4m
\end{eqnarray*}

From Theorem~\ref{thm:one-cluster_bound}, for $s=1$, we obtain the bound
$$(\ell +m-2)(m-1) + (n+1)(m-1)+\ell - 4m\ln\frac{m+3}{2} + 4m.$$
This reaches the maximum for $\ell = n-m$ and is at most
$$2nm-4m\ln\frac{m+3}{2}+2m-n+1.$$
\qed
\end{proof}

In computation, we get the best result applying the formula in Theorem~\ref{thm:one-cluster_bound} with the summand $\sum_{k=1}^{m-1} D^*(m,k)$, which can be easily computed for small $m$. Table~\ref{bounds-tab:D*mk} shows the values for small $m$ and $k$.
In Table~\ref{bounds-tab:sumD*mk}, the values of these summand are given for small non-prime $m$, and the advantages they yield over the bounds in~\cite{Steinberg2011OneClusterPrime,CarpiDAlessandro2013IndependendSetsOfWords}.
To compute the advantages we apply the estimation $m+\ell \leq n$ (with $s=1$ in Theorem~\ref{thm:one-cluster_bound}), and use the fact that the obtained bounds have the same part $2mn + n$ involving $n$. Therefore, the advantages (differences) do not depend on $n$.

\begin{table}
\centering
\caption{The values of $D^*(m,k)$ and $\sum_{k=2}^{m-2} D^*(m,k)$ for non-prime $m$.}\label{bounds-tab:D*mk}
\newcolumntype{C}{>{\centering\let\newline\\\arraybackslash\hspace{0pt}}m{.4cm}}
\begin{tabular}{|c||C|C|C|C|C|C|C|C|C|C|C|C||c|}\hline
$m \backslash k$& 1 & 2& 3& 4& 5& 6& 7& 8& 9&10&11&12 & $\sum D^*(m,k)$\\ \hline \hline
               4& 4 & 3& 4& 1&  &  &  &  &  &  &  &  & 3 \\ \hline
               6& 6 & 5& 4& 5& 6&1 &  &  &  &  &  &  & 14 \\ \hline
               8& 8 & 6& 8& 5& 8&6 &8 & 1&  &  &  &  & 33 \\ \hline
               9& 9 & 9& 7& 9& 9& 7&9 & 9&1 &  &  &  & 50 \\ \hline
              10& 10& 9&10& 9& 6& 9&10& 9&10&1 &  &  & 62 \\ \hline
              12& 12& 9& 8& 7& 8& 7& 8& 7&8 &9 &12& 1& 71 \\ \hline
\end{tabular}
\end{table}

\begin{table}
\centering
\caption{The values of $\sum_{k=1}^{m-1} D^*(m,k)$ and the advantages they yield over other results.}\label{bounds-tab:sumD*mk}
\newcolumntype{C}{>{\centering\let\newline\\\arraybackslash\hspace{0pt}}m{.5cm}}
\begin{tabular}{|l||C|C|C|C|C|C|}\hline
\hspace*{12ex}{$m$}                                                  & 4 & 6 & 8 & 9 & 10 & 12 \\ \hline\hline
$\sum {D}^*(m,k)$                                                    & 11& 26& 49& 68& 82 & 95 \\ \hline
Carpi, D'Alessandro \cite{CarpiDAlessandro2013IndependendSetsOfWords} & 5 & 10& 16& 20& 23 & 29 \\ \hline
Steinberg~\cite{Steinberg2011OneClusterPrime}                        & 9 & 17& 25& 29& 33 & 41 \\ \hline
\end{tabular}
\end{table}


\section{The Algorithm} 

We describe some details of the algorithm that we use in searching for interesting examples of synchronizing automata. Note that the number of automata grows very quickly with the number of letters (see \cite{AGV2013}). To overcome this, we exclude \emph{a priori} a large part of automata from generation process. To do so, we apply the results developed in the previous section.

Such algorithm is useful not only for verifying the \v{C}ern\'{y} conjecture, but for investigating various conjectures in the area.
Our goal is to exhaustively search over the automata with given size $n$ and arity $k$ and to report all those with a long reset length, say with a reset length longer that a predefined $\mathtt{threshold}$. For example, setting $\mathtt{threshold} = (n-1)^2$ verifies the \v{C}ern\'{y} conjecture for the fixed $n$ and $k$.

Obviously, we are interested only in \emph{irreducibly synchronizing} automata, that is those for which removing any letter results with a non-synchronizing automaton. Also, it is well known \cite{Volkov2008Survey} that to verify the \v{C}ern\'{y} conjecture it is enough to consider only \emph{strongly connected} automata (those with the underlying digraph strongly connected), so a special attention is paid to this class.

\subsection{Sieving Procedure}

As we have already mentioned we generate all $k$-ary automata with $n$ states running successively the algorithm described in \cite{KS2013GeneratingSmallAutomata} for arities $i=2,\ldots,k$. For each automaton $\mathcal{A}$ generated by the algorithm, we apply the sieving procedure (Algorithm~\ref{alg:sieving_procedure}). The procedure checks if an automaton should be reported, and whether it should be kept for the next $(i+1)$-th run of the algorithm.

\begin{algorithm}\caption{Sieving Procedure}\label{alg:sieving_procedure}
\begin{algorithmic}[1]
\Require $\mathcal{A}$ -- a generated automaton with $n$ states on $k$ letters.
\Require $\mathtt{threshold}$ -- the bound restricting reset length.
\Procedure{Sieve}{$\mathcal{A}$}
\If{$\mathcal{A}$ is synchronizing}
  \If{$\mathcal{A}$ is strongly connected and irreducibly synchronizing}  
    \State $\ell \gets \text{ the reset length of } A$
    \LineIf{$\ell \ge \mathtt{threshold}$}{Report $\mathcal{A}$}
  \EndIf
\Else
  \State Compute the minimal rank $m$, a word $u$ of this rank, and the set of all compressible pairs $P$.
  \State Compute the bound from successive applications of Theorem~\ref{thm:pairs_bound} and Theorem~\ref{thm:pin_rank_bound} for $\mathcal{A}$.
  \LineIf{the bound is not larger than $\mathtt{threshold}$}{\Return}
  \State $T_\mathcal{A} \gets \text{the transition semigroup of } \mathcal{A}$
  \ForAll{$t \in T_\mathcal{A}$ such that $t$ is a one-cluster transformation}
    \LineIf{the bound from Theorem~\ref{thm:one-cluster_bound} is not larger than $\mathtt{threshold}$}{\Return}
  \EndFor
  \LineIf{$\{t_a\colon a \in \Sigma\}$ is a reducible set of generators of $T_\mathcal{A}$}{\Return}
  \State Store $\mathcal{A}$ for the next run
\EndIf
\EndProcedure
\end{algorithmic}
\end{algorithm}

First we check if $\mathcal{A}$ is synchronizing (line~2). If so, we check if it is irreducibly synchronizing and if its reset length is larger than $\mathtt{threshold}$ (lines~3-6), so we could report it.

If $\mathcal{A}$ is not synchronizing, then it is a potential candidate for further processing in the next run of the algorithm. Using the methods from Section~\ref{sec:bounds} we check if all its irreducibly synchronizing extensions of $\mathcal{A}$ have reset length not larger than $\mathtt{threshold}$.

We compute the minimal rank $m$ of $\mathcal{A}$, a word $u$ of this rank, and the set of all compressible pairs $P$ (line~8). This can be done by a standard BFS algorithm on the power automaton, in the same manner as computing a shortest reset word. It can be seen (as we have mentioned in the remark following Theorem~\ref{thm:pin_rank_bound}) that using the minimal rank and a shortest word of this rank yields the best bound from Theorem~\ref{thm:pairs_bound} and Theorem~\ref{thm:pin_rank_bound}. Hence we apply successively both the propositions for $r=m,\ldots,2$ (line~9).
To use Theorem~\ref{thm:pairs_bound} we also need to know $h(P)$, which is the maximum synchronizing height over the pairs from $P$, and $p(P,m)$, which is the length of an $m$-subset Frankl-Pin sequence over $P$. The value of $h(P)$ is easily computable having $P$. However, there is no known effective algorithm finding a longest $r$-subsets Frankl-Pin sequence over $P$, and a brute-force algorithm has potentially at least double exponential running time. Hence we use the following greedy algorithm to compute such a sequence: We pick a pair $P_i \in P$ whose states are involved in the least number of the other pairs from $P$. We remove $P_i$ from $P$. Then we try to find $M_i \subseteq Q$ in a similar greedy manner. If $M_i$ of the size $r$ is found, then we append $(P_i,M_i)$ to the sequence. We continue this process until $P$ is empty.
If the bound is not larger than $\mathtt{threshold}$, then we skip the automaton (line~10).

For further estimations, we compute the transition semigroup $T_\mathcal{A}$ of $\mathcal{A}$ \cite{FP1997AlgorithmsFiniteSemigroups} (line~11). For each transformation in $T_\mathcal{A}$ we store the length of the shortest words inducing it. Then for each one-cluster transformation we check the bound from Theorem~\ref{thm:one-cluster_bound} (lines~12-14).

Finally, since we are interested only in irreducibly synchronizing automata, we skip also $\mathcal{A}$ if the set of the transformations induced by the letters $a \in \Sigma$ is a reducible set of generators. This is done at the end, since this procedure has high computational cost.

Table~\ref{tab:exclusions_numbers} illustrates savings in computation resulted due to using various exclusions in the sieving procedure. It contains the numbers of automata remaining after exclusions based on the results named in the first column. The second and the third columns contain numbers for the two selected cases of computation with different $\mathtt{threshold}$, $n$, and $k$. Note that the most of exclusions are provided by Theorem~\ref{thm:pairs_bound}.

\begin{table}
\caption{The numbers of non-synchronizing automata (without identity) remaining for the next run in case I) $\mathtt{threshold}=n^2-5n+9$, $n=7$, $k=2$; \ II) $\mathtt{threshold}=(n-1)^2$, $n=6$, $k=3$.}
\centering\label{tab:exclusions_numbers}
\newcolumntype{C}{>{\centering\let\newline\\\arraybackslash\hspace{0pt}}m{.5cm}}
\begin{tabular}{|l||r|r|}\hline
&\multicolumn{1}{c|}{I}&\multicolumn{1}{c|}{II} \\ \hline \hline
No exclusions &\ 7,864,973 &\ 187,138,741 \\ \hline
Irreducible generators of $T_\mathcal{A}$\ & 7,864,331 &\ 179,485,656 \\ \hline
Theorem~\ref{thm:one-cluster_bound} &\ 4,041,171 &\ 74,650,059 \\ \hline
Theorem~\ref{thm:pin_rank_bound} &\ 1,804,727 &\ 3,644,756 \\ \hline
Theorem~\ref{thm:pairs_bound} &\ 1,033,590 & 1,372,878\ \\ \hline
Theorem~\ref{thm:pairs_bound} with Theorem~\ref{thm:pin_rank_bound} &\ 916,354  &\ 1,206,910 \\ \hline
All exclusions &\ 778,517 &\ 515,436 \\ \hline
\end{tabular}
\end{table} 

\subsection{Further Exclusions}

Here we present some simple and technical bounds for the reset length. These are not bounds for the reset lengths of extensions, but are still useful to reduce the number of generated automata.

The following proposition can be used for binary automata ($k=2$), where we can exclude all automata whose letters are of particular form.

\begin{proposition}\label{bounds-pro:idid_bound}
For a binary synchronizing automaton $\mathcal{A}$, if both transformations $t$ generated by the letters satisfy ${t^2}=a$ or $t^2$ is the identity transformation, then $\mathcal{A}$ has a reset word of length at most $2n-2$.
\end{proposition}
\begin{proof}
Let $\Sigma = \{a,b\}$. Consider a shortest reset word $w$. It has no two the same consecutive letters, so it is of the form $w=(ab)^m$ or $w=(ab)^m a$. Obviously, Since $w$ is synchronizing, either $a$ or $ab$ collapses $Q$, that is $|Qab| < |Q|$, and the same can be said for any $Q' = Q(ab)^i$. It follows that $m \le n-1$ and $|w| \le 2n-2$.
\qed
\end{proof}

The following is a well-known folklore result, which ensures that we can restrict ourselves to strongly connected automata and still generate all automata with a long reset length.

\begin{proposition}\label{bounds-pro:non_strongly_connected_bound} Let $\mathcal{A}=\langle Q,\Sigma,\delta \rangle$ be a $k$-ary synchronizing automaton of size $n \ge 5$. If the \v{C}ern\'{y} conjecture is true for all $k$-ary automata of size less than $n$, and $\mathcal{A}$ is not strongly connected, then $\mathcal{A}$ has a reset word of length at most $n^2-4n+5$.
\end{proposition}
\begin{proof}
Let $X \subsetneq Q$ be the sink component of $\mathcal{A}$. There is a word $w$ which maps every state from $Q \setminus X$ to $X$ such that $|w| \le 1+2+\dots+|X| = |X|(|X|+1)/2$. Let $v$ be the shortest reset word synchronizing $X$, $|v| \le (n-1-|X|)^2$. Then $wv$ resets $\mathcal{A}$ and
$$|wv| \le \frac{|X|(|X|+1)}{2}+(n-1-|X|)^2 = n^2 - 2n|X| - 2n + \frac{3}{2}|X|^2 + 5|X| + 1.$$
For $1 \le |X| \le n-1$ and $n \ge 5$ the expression yields the maximum $n^2-4n+5$ for $|X|=1$.
\qed
\end{proof}

We note that in our algorithm we are able to utilize those special results on the \v{C}ern\'y conjecture that refer to a part of the automaton. The results referring to the structure properties not preserved by restrictions of the alphabet are difficult to use. This concerns, in particular, the general result of Grech and Kisielewicz \cite{GK2013AutomataRespectingIntervals}. Nevertheless, we are able to apply successfully a very special case of \cite{GK2013AutomataRespectingIntervals}, which we describe now.

A pair of states $x,y \in Q$ is called \emph{twin}\index{twin pair} in $\mathcal{A}=\langle Q,\Sigma,\delta \rangle$, if $xa=ya$ or $\{xa,ya\} \subseteq \{x,y\}$ for each letter $a\in \Sigma$. In such a case the equivalence relation whose the only nontrivial block is $\{x,y\}$ is a congruence, and the factor automaton is synchronizing if $\mathcal{A}$ is synchronizing. Formally we define 
$\mathcal{A}'=\langle Q',\Sigma,\delta' \rangle$ to be the automaton obtained by identifying the twin states $x$ and $y$ as follows: $Q' = (Q \cup \{z\}) \setminus \{x,y\}$, and
$$\delta'(s,a)=\begin{cases}
z, \text{ if $s \neq z$ and $\delta(s,a) \in \{x,y\}$,} \\
\delta(s,a), \text{ if $s \neq z$ and $\delta(s,a) \in Q \setminus \{x,y\}$,} \\
\delta(x,a), \text{ if $s = z$.}
\end{cases}$$
Then we have
\begin{proposition}\label{bounds-pro:twin_pair_bound}
Let $\mathcal{A}$ be a strongly connected synchronizing automaton. If $\mathcal{A}$  has two twin states $x,y$, and a reset word of length $r$, then the factor automaton $\mathcal{A}'$ is strongly connected, synchronizing, and has a shortest reset word of length $r-1$ or $r$. 
\end{proposition}
\begin{proof}
Any state $p \in Q$ is reachable from $x$, and $x$ is reachable from $p$. Hence the same holds for $p \in Q'$ and $z$, and so, $\mathcal{A}'$ is strongly connected. Obviously, a word synchronizing $\mathcal{A}$ also synchronizes $\mathcal{A}'$.
Let $w$ be a shortest reset word for $\mathcal{A}'$. If $w$ is not a reset word for $\mathcal{A}$, then $Qw=\{x,y\}$. Since $\mathcal{A}$ is synchronizing there exists $a \in \Sigma$ such that $xa = ya$, as otherwise $\{xa,ya\} = \{x,y\}$ for all $a \in \Sigma$. So $wa$ is a reset word for $\mathcal{A}$ of length $|w|+1$.
\qed
\end{proof}

Thus we can skip the automata with a twin pair (which is easily recognizable).
Similarly, as in Proposition~\ref{bounds-pro:non_strongly_connected_bound}, one can observe that if $\mathcal{A}$ is an automaton of size $n$ with a twin pair, and the \v{C}ern\'{y} conjecture is true for all automata of size less than $n$, then $\mathcal{A}$ has reset length at most $n^2-4n+5$.

\section{Results}

We have verified that the \v{C}ern\'{y} conjecture is true for automata with $n \le 5$ states, regardless of the alphabet size. For $n=6$ we checked automata up to $k=5$ letters, for $n=7$ up to $4$ letters, and for $n=8$ up to $3$ letters. Our computation confirms, in particular, all results stated in~\cite{Tr2006Trends}. It is especially important, since it remains unclear whether the results announced in~\cite{Tr2006Trends} are based on partial computation or on unknown ideas. Anyway, no clue is given how and on what base the generation of automata is restricted, while it is apparent that it is impossible to perform this computation with brute force approach. For $n=8$ and $k=3$ we checked 20,933,723,139 automata. Without using parallelism, the computation would take $1.25$ years of one CPU core.
Compare this with the number 572,879,126,392,178,688 of ICDFA automata that one would need to generate applying the technique described in~\cite{AGV2010}.
The algorithm is useful also for verifying many other claims regarding synchronization and discovering new automata with special properties.

\bigskip
\textbf{Acknowledgments.} The computations were performed on a grid that belongs to Institute of Computer Science of Jagiellonian University. We thank Jakub Kowalski for a help in parallelization of our computations, and Jakub Tarnawski for computing the values of $D(m,k)$ for small $m,k$.

\bibliographystyle{plain}

\end{document}